\documentclass[letterpaper,10pt,conference]{ieeeconf}
\IEEEoverridecommandlockouts

\usepackage{cite}
\usepackage{float}
\usepackage{amsmath,amssymb,amsfonts}
\usepackage{mathrsfs}
\usepackage{graphicx}
\usepackage{textcomp}
\usepackage[ruled,vlined]{algorithm2e}
\usepackage{mathtools}
\usepackage{adjustbox}
\usepackage{makecell}
\usepackage{empheq} 
\usepackage{xcolor}
\usepackage[amsmath,thmmarks]{ntheorem}
\usepackage{hyperref}
\usepackage{enumerate}
\hypersetup{pdfborder={0 0 0},}
\def\BibTeX{{\rm B\kern-.05em{\sc i\kern-.025em b}\kern-.08em
    T\kern-.1667em\inner.7ex\hbox{E}\kern-.125emX}}
 \theoremheaderfont{\normalfont\bfseries} % 标题加粗
\theorembodyfont{\normalfont\itshape}    % 正文斜体
\theoremseparator{.}                     % 编号后添加句点

\newtheorem{remark}{Remark}
\newtheorem{theorem}{Theorem}
\newtheorem{lemma}{Lemma}

\newtheorem{assumption}{Assumption}
\usepackage{graphicx}
\usepackage{textcomp}
\usepackage{color}

\begin{document}
\title{\LARGE \bf A Global Convergence Analysis of Consensus ALADIN for Convex Optimization}
\author{Xu Du,  Shuting Wu,  Karl~H.~Johansson, and Apostolos I. Rikos$^*$
\thanks{$^*$Corresponding author.}
	\thanks{Xu Du and Apostolos I. Rikos are with the Artificial Intelligence Thrust of the Information Hub, The Hong Kong University of Science and Technology (Guangzhou), Guangzhou, China. 
    Apostolos I. Rikos is also affiliated with the Department of Computer Science and Engineering, The Hong Kong University of Science and Technology, Clear Water Bay, Hong Kong, China. E-mails: {\tt~\{michaelxudu, apostolosr\}@hkust-gz.edu.cn}. 
            }
            \thanks{Shuting Wu is with School of Mathematics and Statistics, North China University of Water Resources and Electric Power, Zhengzhou, China.  E-mail: \texttt{wushuting0126@163.com}.}
           % \thanks{Fumin Zhang is with the Department of Electronic and Computer Engineering, Hong Kong University of Science and Technology, Hong Kong SAR, China. E-mail: {\tt eefumin@ust.hk}.
            %} 
            \thanks{Karl H.~Johansson is with the Division of Decision and Control Systems, KTH Royal Institute of Technology, SE-100 44 Stockholm, Sweden. 
    He is also affiliated with Digital Futures, SE-100 44 Stockholm, Sweden. 
    E-mail:{\tt~kallej@kth.se}.
            }
            \thanks{The work of X.D. and A.I.R. was supported by the Guangzhou-HKUST(GZ) Joint Funding Scheme (Grant No. 2025A03J3960). The work of A.I.R. was also supported by the Guangdong Provincial Project (Grant No. 2024QN11G109).}
}

\maketitle

\begin{abstract}
Distributed optimization problems are pervasive in machine learning and optimal control. 
In this paper, we study smooth strongly convex distributed consensus optimization problems. We present a distributed optimization algorithm for consensus problems based on the Consensus Augmented Lagrangian Alternating Direction Inexact Newton (C-ALADIN) framework.
Our algorithm uses an auxiliary variable to decide when to update second-order information, enabling curvature exploitation without sacrificing global convergence. This contrasts with existing C-ALADIN methods, which require constant Hessian approximations and thus lose numerical advantages. Under smooth strong convexity, the algorithm converges globally, and the auxiliary variable converges sublinearly. Numerical experiments on logistic regression show that our algorithm outperforms baseline methods that use either fixed or updated Hessian information.
\end{abstract}
\textbf{Keywords:} Distributed Convex Optimization, Consensus ALADIN, Global Convergence    
%\end{IEEEkeywords}

\section{Introduction}\label{sec:introduction}
Distributed optimization  has received significant attention in recent years due to its wide range of applications, including machine learning \cite{zhou2026preconditioned}, \cite{ALADINlearning}, model predictive control (MPC) and moving horizon estimation (MHE) \cite{stomberg2025decentralized, wu2025time}, power grids \cite{dai2026distributed,engelmann2022distributed}, and signal processing \cite{boyd2011distributed},
\cite{rikos2023asynchronous}. 
Overall, distributed optimization enables efficient large-scale problem solving with improved scalability and data privacy \cite{2024_doostmohammadian_rikos_Johansson_survey}. %\todo{... fill one more} 

%\todo{this paragraph, please write better. why mention DRA if we do not do DRA? \\
%talk about consensus distributed optimization and split consensus distributed optimization into 2 parts -- primal decomposition and dual. and we do the latter}
Distributed optimization problems are often formulated with a finite-sum objective. Two common problem types are: (i) distributed resource allocation optimization problems, where local variables are affinely coupled; (ii) distributed consensus optimization (DCO) problems, where all local variables are constrained to be equal. This paper focuses on the latter. DCO algorithms typically involve two steps: (i) solving local subproblems, and (ii) exchanging information with a central coordinator or neighboring agents. Structurally, they can be classified into primal decomposition and dual decomposition methods \cite{decomposition}. This paper mainly focuses on the latter, in particular the Augmented Lagrangian based Alternating Direction Inexact Newton (ALADIN) method \cite{houska2016augmented, Du2025ACC},
 which exhibits fast convergence in many distributed optimization applications \cite{ALADIN-tool}.  ALADIN was originally proposed to enhance the convergence speed of the alternating direction method of multipliers (ADMM) \cite{boyd2011distributed, rikos2023asynchronous, du2026affine} and to achieve more robust convergence in non-convex applications, specifically for distributed resource allocation optimization problems. 
It inherits the distributed structure of ADMM and the fast convergence properties of sequential quadratic programming (SQP). \\
%which has shown excellent \todo{another word? ''excellent'' does not give information. ''fast'' maybe? or ''fast numerical performance against other algorithms ... etc in many} numerical performance in many distributed optimization applications \cite{ALADIN-tool}.\\
% Distributed optimization problems are mainly divided into two types based on the finite-sum objective function structure: (i) distributed resource allocation problems, where the variables of local agents are affinely coupled; and (ii) distributed consensus (DCO) optimization problems, where the variables are subject to consensus constraints. This paper focuses on the latter. Distributed optimization algorithms typically involve two steps: (i) solving local subproblems, and (ii) exchanging information with a central coordinator or neighboring agents. Structurally, they can be classified into primal decomposition and dual decomposition methods \cite{decomposition}. This paper mainly focuses on the latter, in particular the Augmented Lagrangian based Alternating Direction Inexact Newton (ALADIN) method \cite{houska2016augmented}, \cite{Du2025ACC}.\\
\textbf{Literature Review.} %\todo{in this part, we did not say why ''enabling Hessian updates is important'', we say that previous works do not allow it, but reviewers may ask why allowing this is important?}
ALADIN has been adopted in a wide range of research directions, including (i) power grids via spatial splitting \cite{Engelmann2019, dai2026distributed, Du2019, lanzadistributed}, (ii) MPC and MHE via time splitting \cite{Kouzoupis2016, stomberg2025decentralized, wu2025time, wang2026lightweight}, and (iii) structure splitting \cite{Shi2022, wang2025aladin}, where distributed solution methods help mitigate the curse of dimensionality. 
Convergence results for ALADIN can be summarized as follows. For non-convex problems satisfying the second-order sufficient condition (SOSC), Hessian updates are admissible, and ALADIN inherits local convergence from SQP \cite{houska2016augmented}. For convex problems, the existing ALADIN theory guarantees global convergence via a Lyapunov function that depends on constant Hessian approximations. 
Constant Hessians, however, cause ALADIN to lose the numerical acceleration of second-order information \cite{houska2017convex}.
Updating the Hessian, on the other hand, prevents the construction of the required Lyapunov function, and hence the global convergence theory of ALADIN cannot be established. Therefore, it is essential to design a method that allows Hessian updates without sacrificing the global convergence of ALADIN. 
Moreover, standard second-order methods typically require damping techniques (line search or trust region) for global convergence, which cause significant overhead and complicate coordination in ALADIN-type algorithms. As discussed in \cite[Section~4]{houska2017convex}, it is crucial to design a method that combines Hessian updates with global convergence to achieve faster numerical convergence without sacrificing global guarantees.
To date, no existing work within the ALADIN framework simultaneously addresses two challenges: (i) a structure specifically designed for DCO problems, and (ii) updating Hessian information while preserving global convergence. The first challenge was addressed by Consensus ALADIN (C-ALADIN) \cite{du2023consensus, Du2025ACC, Du2025ECC, du2025decentralized, han2026mix}, which replaces the coupled quadratic program (QP) in ALADIN with a consensus QP, thereby directly handling the consensus constraints.
Regarding the second challenge, to the author's knowledge, it remains unresolved. The only existing work is a heuristic attempt in \cite[Section~4.2]{houska2017convex}, which proposes a triggering condition for Hessian updates based on an $L_1$ merit function and heuristic parameters. However, this approach lacks rigorous convergence analysis and is not tailored to DCO problems. Consequently, no ALADIN-type framework currently permits Hessian updates while retaining global convergence guarantees for DCO problems.
% \AR{Regarding the second challenge ... to the author's knowledge ... only \cite[Section~4.2]{houska2017convex}.  
% In this work a triggering condition for Hessian updates based on an $L_1$ merit function and heuristic parameters was proposed in the convex setting. 
% However, this approach lacks rigorous convergence analysis for the optimization variables and is not designed for DCO problems. 
% As a result, to the best of our knowledge, no ALADIN‑type framework permits Hessian updates while retaining global convergence guarantees for DCO problems.} 
% \todo{The second challenge, however, \AR{to the author's knowledge} remains unresolved. \todo{we say it remains unresolved and then we present an approach?}
% To address it, a triggering condition for Hessian updates based on an $L_1$ merit function and heuristic parameters was proposed in the convex setting \cite[Section~4.2]{houska2017convex}. 
% Yet this approach lacks rigorous convergence analysis for the optimization variables and is not designed for DCO problems. 
% Consequently, to the best of our knowledge, no ALADIN‑type framework permits Hessian updates while retaining global convergence guarantees for DCO problems.} 
\\
% \todo{
% Please emphasize more why filling this gap is important
% }
%\todo{this paragraph please re-write again. the beginning of this paragraph is not literature review. Literature review needs to say a story of the problems/gaps that other works solved. End up at the most recent works, and them mention the gap. And emphasize on why adressing the gap is important.}
\textbf{Main Contributions.} Motivated by the aforementioned limitations in the existing literature, this paper proposes a C-ALADIN-based algorithmic framework. 
To the best of our knowledge, this is the first C-ALADIN framework that guarantees global convergence while allowing Hessian updates. 
Our main contributions are summarized as follows.\\
\textbf{A.} 
We propose CAPTAIN (C-ALADIN with parameter tuning and adaptive inexact Newton), a novel distributed optimization algorithm, to solve DCO problems (Algorithm~\ref{alg: AC-ALADIN}).  In CAPTAIN, an auxiliary variable governs the timing of Hessian updates. In contrast to the heuristic condition in \cite[Section~4.2]{houska2017convex}, our update rule relies on the global objective and an adaptive threshold. For smooth and strongly convex DCO problems, we establish global convergence by showing that the auxiliary variable is updated only finitely many times (Theorem~\ref{thm:z_behavior}) and that the minimum distance between consecutive updates decays sublinearly (Theorem~\ref{them: sublinear}). Once the auxiliary variable stops being updated, the Hessian approximations remain fixed, and CAPTAIN reduces to the Hessian-invariant variant of C-ALADIN, which is globally convergent (Lemma~\ref{lemma 1}). This guarantees global convergence of CAPTAIN. \\
\textbf{B.} We evaluate the numerical performance of CAPTAIN on a smooth strongly convex $\ell_2$-regularized logistic regression application. 
Compared with other state-of-the-art DCO algorithms, CAPTAIN exhibits faster convergence, reaching the optimal solution in significantly fewer iterations. %We evaluate the numerical performance of CAPTAIN on a smooth strongly convex $\ell_2$-regularized logistic regression application. When compared with other state-of-the-art DCO algorithms, CAPTAIN shows a distinct advantage in speed, consistently reaching the optimal solution in significantly fewer iterations. The objective function decreases sharply from the very first steps, and the algorithm rapidly attains the optimal value. 

\textit{Notation.} In this paper, $(\cdot)^{[k]}$ denotes the value at the $k$-th iteration of the algorithm. For a vector $\xi \in \mathbb{R}^n$, the symbols $\|\xi\|$ and $\|\xi\|_\infty$ denote the Euclidean norm and the infinity norm, respectively. Moreover, for $M \succ 0$, $\|\xi\|_{M}^2 = \xi^{\top} M \xi$ represents the squared Mahalanobis distance. Furthermore, $|\mathcal{S}|$ indicates the cardinality of a countable set $\mathcal{S}$.

\section{Problem Formulation}\label{sec: problem}
% \todo{write what we want to do in this paper. \\
% In this paper we aim to design an algorithm that solves... Furthermore, our algorithm will ... (characteristics of the algorithm). 
% }

%\todo{this above paragraph after the (2)... and also refer to (2)}

Let $\mathcal{V}$ denote the set of agents. A DCO problem with $N = |\mathcal{V}|$ agents can be formulated as
\begin{equation}\label{eq: merit}\small
\begin{aligned}
   \min_{y \in \mathbb R^{n}} \quad &  \Phi(y) \doteq \sum_{i=1}^{N} f_i(y),
\end{aligned}
\end{equation}
where $f_i(\cdot): \mathbb R^n \rightarrow \mathbb R$ is the local objective function of agent $i \in \mathcal V$.
To enable distributed computation, we reformulate \eqref{eq: merit} by introducing local variables $x_i$ for each agent:
\begin{equation}\label{eq: DC}\small
    \begin{aligned}
        \min_{x_i\in \mathbb R^n,\, i \in \mathcal V} \quad  \sum_{i=1}^{N} f_i(x_i) \qquad
        \text{s.t.}\;\;  x_i = y, \quad \forall i \in \mathcal V.
    \end{aligned}
\end{equation}
%\todo{Assumption 1 before section IV-A} 
In this paper we aim to design an algorithm that solves the DCO problem in \eqref{eq: DC} with guaranteed global convergence.
% Furthermore, our algorithm will provide global convergence guarantees while allowing adaptive Hessian updates for convex problems, a capability missing in existing C-ALADIN frameworks.}

\section{Preliminaries on Consensus ALADIN}\label{sec: C-ALADIN}
The augmented Lagrangian function of \eqref{eq: DC} can be formulated as
\begin{equation}\label{eq: AL}%\small
    \mathcal L (x, y, \lambda) 
    \doteq \sum_{i=1}^N\left( 
    f_i(x_i) 
    + \lambda_i^\top (x_i - y ) 
    + \frac{1}{2}\left\|x_i - y\right\|_{M_i}^2
    \right), 
\end{equation}
where $x = [x_1^\top, x_2^\top, \dots, x_N^\top ]^\top$ collects the local primal variables, $\lambda = [\lambda_1^\top, \lambda_2^\top, \dots, \lambda_N^\top]^\top$ collects the dual variables, and $M_i \succ 0$ is a symmetric positive definite matrix for each agent.
Based on \eqref{eq: AL}, we first recall the derivative-free variant of C-ALADIN, termed DFC-ALADIN \cite{Du2025ACC}, which is designed for convex DCO problems and keeps the Hessian approximation $M_i$ constant. In contrast, when the matrices $M_i$ are updated as approximate Hessians for every agent $i$ at each iteration, we refer to the algorithm as standard C-ALADIN.
The DFC-ALADIN iteration is given by:
\begin{subequations}\label{eq: CALADIN}\small
\begin{align}
x_i^{[k+1]} &
= \mathop{\arg\min}_{x_i, 
\forall i \in \mathcal V, } 
\; f_i(x_i) 
+ \left(\lambda_i^{[k]}\right)^\top x_i 
+ \frac{1}{2}\left\|x_i - y^{[k]}\right\|_{M_i}^2,\label{eq: local x update}\\
g_i & = M_i\left(y^{[k]} - x_i^{[k+1]} \right) - \lambda_i^{[k]}, \label{eq: local g update} \\
y^{[k+1]} &= \left(\sum_{i=1}^N M_i \right)^{-1} \left(\sum_{i=1}^N M_i x_i^{[k+1]} - g_i \right), \label{eq: y update} \\
\lambda_i^{[k+1]} 
&= M_i\left(x_i^{[k+1]} - y^{[k+1]}\right) -g_i, \label{eq: dual update}
\end{align}
\end{subequations}
for every agent $i \in \mathcal V$, with $k$ denoting the iteration index.
% The algorithm consists of four main steps. 
% In \eqref{eq: local x update}, each agent updates its local primal variable $x_i^{[k+1]}$. Equation \eqref{eq: local g update} evaluates the local gradient.
% The global variable $y^{[t+1]}$ is updated via \eqref{eq: y update}, and the dual variables $\lambda_i$ are updated via \eqref{eq: dual update}. 
Note that, the updates of \eqref{eq: y update} and \eqref{eq: dual update} are the closed-form solutions of the following consensus QP problem, which performs the coordination in C-ALADIN:
\begin{equation}\label{eq: consensus QP}\small
    \begin{aligned}
        \min_{x_i,\, i \in \mathcal V} \quad & \sum_{i=1}^{N} \frac{1}{2} \Delta x_i^\top M_i \Delta x_i + g_i^\top \Delta x_i \\
        \text{s.t.}\quad\;\; & x_i^{[k+1]} + \Delta x_i = y, \quad \forall i \in \mathcal V. 
    \end{aligned}
\end{equation}
As shown in \cite{Du2025ACC, Du2025ECC, du2025decentralized}, for smooth non-convex problems satisfying SOSC, standard C-ALADIN achieves local convergence. In contrast, for convex problems, global convergence is guaranteed only for DFC-ALADIN, i.e., when $M_i$ is fixed. In this paper, we propose an algorithm that allows Hessian updates while preserving global convergence for convex DCO.

\section{The Proposed Algorithm: CAPTAIN}\label{sec: ALgorithm}
In this section, we propose CAPTAIN, a novel distributed numerical algorithm for solving \eqref{eq: DC}. The algorithm is presented as Algorithm~\ref{alg: AC-ALADIN} in Section~\ref{sec: GC-ALADIN}, and its global convergence is established in Section~\ref{sec: convergence}.

For the subsequent algorithm design and convergence analysis, we first make the following assumption, which is fundamental for the rest of the paper.
\begin{assumption}\label{ass}
For each $i \in \mathcal{V}$, let the objective function $f_i$ be closed, proper, smooth, strongly convex, and bounded from below. Moreover, problem \eqref{eq: DC} is feasible and strong duality holds. Specifically, there exists a strong convexity constant $\mu_i > 0$ such that for all $x_\alpha, x_\beta \in \mathbb{R}^n$,
\begin{equation}\label{eq: mu}
    f_i(x_\alpha) + \nabla f_i(x_\alpha)^\top (x_\beta - x_\alpha) + \frac{\mu_i}{2}\|x_\beta - x_\alpha\|^2 \leq f_i(x_\beta).
\end{equation}
\end{assumption}
Assumption~\ref{ass} is necessary for the design of our algorithm and the establishment of its global convergence. Strong convexity and smoothness guarantee the uniqueness of the optimal solution to \eqref{eq: DC}. Furthermore, boundedness from below supports the convergence analysis and ensures finite-time termination.

\subsection{Algorithm Development}\label{sec: GC-ALADIN}

The proposed CAPTAIN is presented as Algorithm~\ref{alg: AC-ALADIN} and illustrated in Fig.~\ref{fig: flow}. 
\begin{algorithm}[ht] \small
	\caption{CAPTAIN: Consensus ALADIN with Parameter Tuning and Adaptive Inexact Newton}
    \textbf{Initialization.}  Chosen dual variable $ \lambda_i^{[1]} \in \mathbb{R}^n$, and global variable estimation $ z^{[1]} \in \mathbb{R}^n$, $ y^{[1]} = z^{[1]} $, for each node $i \in \mathcal{V}$.  Set $M_i = I$, $\gamma^{[1]} = N$ and $\epsilon>0$. Assumption \ref{ass} holds.
    \\
    \textbf{Iteration.} 
	\begin{enumerate}
	\item  Optimize $x_i^{[k+1]}$ as \eqref{eq: local x update}.
    \item Evaluate the gradient $g_i$ of $f_i$ as \eqref{eq: local g update}.

    \item  Calculate the global variable estimation as  \eqref{eq: y update}.

    \item Update the dual variable $\lambda_i$ as \eqref{eq: dual update}.

    \item If $\Phi\left(y^{[k+1]}\right)< \Phi\left(z^{[t]}\right) - \frac{\gamma^{[t]}}{2}\left\|y^{[k+1]} -  z^{[t]}   \right\|^2$ and $\left\| y^{[k+1]} -  z^{[t]}\right\|\geq \epsilon$,
    set  $z^{[t+1]} = y^{[k+1]}$ and update the Hessian matrices
    \begin{equation*}\label{eq: hesian approximation}
        M_i \approx \nabla^2 f_i\left(z^{[t+1]}\right) \succ \mu_i I\succ0,\; \forall i\in \mathcal V.
    \end{equation*} Set 
     $\gamma^{[t+1]} = \rho_{\text{min}}\left(\sum_{i=1}^N M_i\right) $, where  $\rho_{\text{min}}(\cdot)$ denotes the smallest eigenvalue. 
	\end{enumerate}
    \textbf{Output.} Optimal solution $y^* = \lim_{k\to\infty} y^{[k]}$ of \eqref{eq: DC}.
	\label{alg: AC-ALADIN}
\end{algorithm}
The algorithm consists of two phases: (i) steps 1) – 4) constitute the optimization phase, which are the same as those of DFC-ALADIN (see \eqref{eq: CALADIN}); and (ii) step 5) updates the auxiliary variable \(z^{[t+1]} = y^{[k+1]}\), the local Hessian approximations \(M_i\), and the parameter \(\gamma^{[t]}\) whenever both of the following conditions hold:
\begin{equation}\label{eq: decrease}\small
\left\{\begin{split}
 &   \Phi\left(y^{[k+1]}\right) < \Phi\left(z^{[t]}\right) - \frac{\gamma^{[t]}}{2}\left\|y^{[k+1]} - z^{[t]}\right\|^2,\\
  &  \left\| y^{[k+1]} - z^{[t]}\right\| \geq \epsilon.
\end{split} \right.
\end{equation} %The only difference between \eqref{eq: global y} in step 3) and \eqref{eq: dual update} is that \eqref{eq: global y} denotes the closed-form solution of the following consensus QP:
% \begin{equation}\label{eq: consensus QP}
%     \begin{aligned}
%         \min_{x_i,\, i \in \mathcal V} \quad & \sum_{i=1}^{N} \frac{1}{2} \Delta x_i^\top M_i \Delta x_i + g_i^\top \Delta x_i + \frac{\gamma^{[t]}}{2}\left\|y -y^{[k]}\right\|^2 \\
%         \text{s.t.}\quad\;\; & x_i^{[k+1]} + \Delta x_i = y, \quad \forall i \in \mathcal V, 
%     \end{aligned}
% \end{equation}
% which ensures stable numerical updates of the global variable $y$.
%  In detail, step 5) updates the auxiliary variable \(z^{[t+1]} = y^{[k+1]}\), the local Hessian approximations \(M_i\), and the parameter \(\gamma^{[t]}\) whenever both of the following conditions hold:
% \begin{equation}\label{eq: decrease}\small
% \left\{\begin{split}
%  &   \Phi\left(y^{[k+1]}\right) < \Phi\left(z^{[t]}\right) - \frac{\gamma^{[t]}}{2}\left\|y^{[k+1]} - z^{[t]}\right\|^2,\\
%   &  \left\| y^{[k+1]} - z^{[t]}\right\| \geq \epsilon.
% \end{split} \right.
% \end{equation}
% and
% \begin{equation}\label{eq: decrease2}\small
%     \left\| y^{[k+1]} - z^{[t]}\right\| \geq \epsilon.
% \end{equation}
Therefore, Hessian information is updated only when a sufficient decrease is achieved and the step size is non‑negligible, ensuring efficient use of second‑order information.
Moreover, the adaptive choice of \(\gamma^{[t]}\) directly governs the update of \(z^{[t]}\). Setting \(\gamma^{[t]} = \rho_{\min}\bigl(\sum_i M_i\bigr) > 0\) ensures a strictly positive threshold, guaranteeing sufficient descent in the global objective for each accepted update.
This adaptive mechanism of $\gamma$ has two effects: 
\begin{itemize}
    \item When the Hessian approximations have large eigenvalues (i.e., provide strong curvature), \(\gamma^{[t]}\) is large, preventing overly aggressive updates of \(z^{[t]}\) and stabilizing the iterations.
    \item When the Hessian approximations have small eigenvalues (i.e., provide weak curvature), \(\gamma^{[t]}\) is small, allowing \(z^{[t]}\) to follow \(y^{[k+1]}\) more quickly and accelerating convergence.
\end{itemize}
\begin{figure}[ht]
	\centering
\includegraphics[width=0.5\textwidth,height=0.46\textheight]{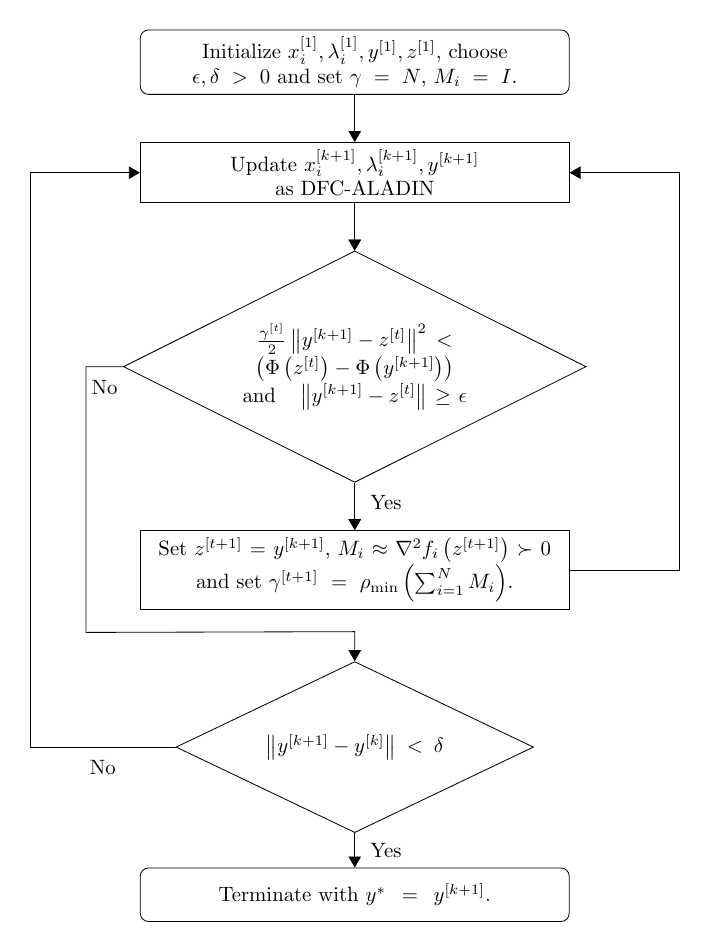}
	\caption{Flow chart of CAPTAIN.}
	\label{fig: flow}
\end{figure}
In this way, \(\gamma^{[t]}\) automatically balances the descent requirement and update speed, leading to monotone convergence of \(z^{[t]}\), which in turn ensures a monotonic decrease of the global objective value in \eqref{eq: merit}.\\
\textbf{Comparison with Previous Works.} As shown in Section~\ref{sec: C-ALADIN}, when problem \eqref{eq: DC} satisfies SOSC and the Hessian matrices $M_i$ are updated at every iteration, C-ALADIN \cite{Du2025ACC} yields only local convergence; global convergence theory is missing. 
Furthermore, DFC-ALADIN \eqref{eq: CALADIN} requires a constant Hessian $M_i\succ 0$ for each agent to guarantee global convergence. In contrast, Algorithm~\ref{alg: AC-ALADIN} maintains global convergence while adaptively updating Hessian approximations, which leads to faster numerical convergence, as demonstrated in the next section.
A similar triggering condition for Hessian updates was proposed in \cite[Section~4.2]{houska2017convex} for resource-allocation ALADIN, relying on an $L_1$ merit function and heuristic parameters. However, it lacks a convergence proof for the optimization variables. The same update strategy is applied to power grids in \cite{dai2026distributed}; both works address distributed resource allocation, not DCO. Other dual-decomposition algorithms (e.g., C-ADMM \cite{boyd2011distributed}) do not exploit second-order information and share the convergence theory of DFC-ALADIN. %Numerical comparisons of the aforementioned methods are presented in Section~\ref{sec: numerical}. %Therefore, Algorithm~\ref{alg: AC-ALADIN}   improves numerical performance.

 % \todo{can we mention more works here?} 

\subsection{Global Convergence Analysis}\label{sec: convergence}
The global convergence of Algorithm~\ref{alg: AC-ALADIN} is established in this section. We first introduce a lemma that serves as the foundation for the subsequent convergence results.

\begin{lemma}\label{lemma 1}
   Let Assumption~\ref{ass}  be satisfied for \eqref{eq: DC}. Then the convex problem \eqref{eq: DC}, when solved by DFC-ALADIN (see \eqref{eq: CALADIN}), converges globally with linear rate.
\end{lemma}
\begin{proof}
    See \cite[IV.A]{Du2025ACC} and \cite[Theorem 3]{du2023consensus}.
\end{proof}

The following theorem formally states the main convergence result for the auxiliary variable $z$ in Algorithm~\ref{alg: AC-ALADIN}.
\begin{theorem}%[Convergence of the auxiliary variable $z$]
\label{thm:z_behavior}
Let Assumption~\ref{ass} hold for problem~\eqref{eq: DC}, and let the auxiliary variable $z^{[t]}$ in Algorithm~\ref{alg: AC-ALADIN} be updated according to the sufficient decrease condition \eqref{eq: decrease} with 
\(\gamma^{[t]} = \rho_{\min}\left(\sum_{i=1}^N M_i\right) > 0\). 
Then the following statements hold:
\begin{enumerate}
    \item The sequence of auxiliary variables $\{z^{[t]}\}$ is updated only finitely many times.
    \item At its final update, the auxiliary variable $z^{[T]}$ either coincides with the global optimum $y^*$, or satisfies $\left\|y^{[k+1]}-z^{[T]}\right\|<\epsilon$ (which violates the second condition of \eqref{eq: decrease}), or reaches a point where any further update would violate the first condition of \eqref{eq: decrease}. In the latter case, the distance to the optimum satisfies
\begin{equation}\label{eq:z_neighborhood}\small
    \left\| z^{[T]} - y^* \right\| \geq \sqrt{\frac{2}{\gamma^{[T]}} \bigl( \Phi(z^{[T]}) - \Phi(y^*) \bigr) }.
\end{equation}
\end{enumerate}
\end{theorem}

\begin{proof} See Appendix \ref{app: z_behavior}.     
\end{proof}

\begin{remark}
This result shows that the auxiliary variable sequence $\{z^{[t]}\}$ undergoes only finitely many updates and does not necessarily converge to the optimal solution $y^*$, due to the stopping criterion \eqref{eq: decrease}.
By construction of Algorithm~\ref{alg: AC-ALADIN}, the Hessian approximation matrices $\{M_i\}$ are updated only at iterations when $z^{[t]}$ is updated. Therefore, once $z^{[t]}$ ceases to update at iteration $T$, the matrices $\{M_i\}$ remain fixed thereafter.
In this regime, Algorithm~\ref{alg: AC-ALADIN} reduces to DFC-ALADIN. Since the global convergence of DFC-ALADIN holds for any initial point (see Lemma~\ref{lemma 1}), the sequences $\{y^{[k]}\}$ and $\{x_i^{[k]}\}$ converge to the global minimizer $y^*$ for all $i \in \mathcal V$.
\end{remark}

The following theorem provides the convergence rate analysis of the auxiliary variable $z$.
\begin{theorem}%[Sublinear Decay Rate of the Minimum Update Magnitude of $z$]
\label{them: sublinear}
Suppose that Assumption~\ref{ass} holds for problem~\eqref{eq: DC}, and let problem~\eqref{eq: DC} be solved by Algorithm~\ref{alg: AC-ALADIN}. Then the minimum distance between consecutive updates of the auxiliary variable $z$ decays sublinearly, i.e.,
  \begin{equation}\label{eq: decrease theorem}\small
  \min_{1\leq t\leq T-1}\left\|z^{[t+1]} - z^{[t]}\right\| = O\left(\frac{1}{\sqrt{T}}\right),
\end{equation}
where $T$ denotes the total number of updates of the sequence $\{z^{[t]}\}$.
\end{theorem}

\begin{proof} See Appendix \ref{app: sublinear}.
\end{proof}
Meanwhile, the convergence rates of $x_i$ and $y$ in Algorithm~\ref{alg: AC-ALADIN} during the intervals where $z$ remains unchanged are linear, as established in \cite{Du2025ACC}, and are omitted here for brevity. 
\section{Numerical Application:  Logistic Regression Problem}\label{sec: numerical}

%{\color{blue}Compare with
%Derivative-free C-ALADIN, RC-ALADIN, BFGS Consensus ALADIN \cite{Du2025ACC}, Consensus ADMM \cite{Boyd2011},  Linearized Consensus ADMM \cite{wu2025cocoa}, DQM\cite{DQM}, GIANT \cite{wang2018giant},     %GD \cite{nesterov2018lectures}

%Numerical example %\cite{kan2024lsemink}, %\url{https://github.com/KelvinKan/LSEMINK?utm_source=pro.scigpt.club}
%\cite{kovalev2020optimal}}

% All simulations are conducted using \texttt{Casadi-3.6.6} with \texttt{IPOPT} in \texttt{MATLAB R2024a} on a Windows $11$ system, equipped with a $2.1$ GHz AMD Ryzen $5$ $4600$U processor and $16$ GB of RAM.
In this section, we present numerical simulations to illustrate the performance of Algorithm~\ref{alg: AC-ALADIN} and to highlight its advantages over existing distributed optimization methods. We consider a distributed $\ell_2$-regularized logistic regression problem, following the experimental setup of \cite{kovalev2020optimal}. The regularized logistic regression problem satisfies Assumption~\ref{ass}, thereby providing a unified testbed for both theoretical validation and numerical evaluation.%Logistic regression is a standard benchmark in distributed optimization, and its regularized form guarantees that each local objective satisfies the strong convexity required by Assumption~\ref{ass}, thereby providing a unified testbed for both theoretical validation and numerical evaluation. \todo{compress last sentence -- verbose}

 The numerical experiments are based on the \texttt{LIBSVM} dataset \texttt{ijcnn1.tr}. From this dataset, $10,000$ samples are randomly selected and evenly distributed among $N = 100$ agents, with each agent holding $m = 100$ samples. For each agent $i$, let $a_{ij} \in \mathbb{R}^n$ denote the feature vector of the $j$-th sample and $b_{ij} \in \{-1, 1\}$ its corresponding label. The local data matrix $A_i \in \mathbb{R}^{m \times n}$ is defined as $A_i = [a_{i1}, a_{i2}, \dots, a_{im}]^\top$. The local objective function at agent $i$ is given by
\begin{equation} \label{eq:l2_logistic}\small
f_i(y) = \frac{1}{m} \sum_{j=1}^{m} \log\bigl(1 + \exp(-b_{ij} a_{ij}^\top y)\bigr) + \frac{r}{2}\|y\|^2,
\end{equation}
where $r = 10^{-3}$ is the regularization parameter and $y \in \mathbb{R}^n$ is the decision variable. For the \texttt{ijcnn1.tr} dataset, the feature dimension is $n = 22$.
%In our implementation,  %{\color{blue} (e.g., $d=22$ for ijcnn1, $d=122$ for a6a, and $d=300$ for w6a).}
All simulations are performed using \texttt{CasADi 3.6.6} with the \texttt{IPOPT} solver in \texttt{MATLAB R2024a} on a Windows 11 system equipped with a 2.1 GHz AMD Ryzen 5 4600U processor and 16 GB of RAM. Both the optimal solution $y^*$ and the optimal objective value $\Phi(y^*)$ are computed via \texttt{CasADi}, where we solve the corresponding centralized problem to obtain a high‑accuracy reference solution. These serve as benchmarks for evaluating the performance of the distributed algorithms. 
% We evaluate the convergence performance of the proposed Algorithm \ref{alg: AC-ALADIN}, Algorithm~\ref{alg: AC-ALADIN} with \eqref{eq: hessian}and Algorithm~\ref{alg: AC-ALADIN} with BFGS against DFC-ALADIN\cite{Du2025ACC}, RC-ALADIN\cite{Du2025ACC}, C-ADMM\cite{boyd2011distributed}, DQM\cite{DQM}, and GIANT\cite{wang2018giant} on the \texttt{ijcnn1} dataset.
% All algorithms are initialized from the same 
% starting point $y^{[1]} = \mathbf{0}$, with dual variables 
% $\lambda_i^{[1]} = \mathbf{0}$. For C-ADMM and DQM, 
% the penalty parameter is fixed to $\rho = 1$, respectively, and both results are reported; since DQM is 
% inherently decentralized, we implement a coordinator-based version for fair 
% comparison. For GIANT, the approximation parameter is set to 
% $\epsilod = 1$. For DFC-ALADIN, the Hessian approximation is 
% fixed as
% \begin{equation} \label{eq: hessian}
%     M_i = \mathrm{diag}\!\left(\frac{1}{4m}A_i^\top A_i\right) + rI + \varepsilon I, 
%     \quad \varepsilod = 10^{-4},
% \end{equation}
% where $A_i \in \mathbb{R}^{m\times d}$ is the local data matrix of agent 
% $i$, and the factor $\frac{1}{4}$ arises from the global upper bound 
% $\sigma(1-\sigma)\leq \frac{1}{4}$ on the logistic loss Hessian. For 
% RC-ALADIN, $M_i = I$. For the proposed CAPTAIN, we 
% initialize $M_i = I$, $\gamma^{[1]} = n$, and $z^{[1]} = \mathbf{0}$. 

\begin{figure}[htbp]
    \centering
\includegraphics[width=0.46\textwidth,height=0.32\textheight]{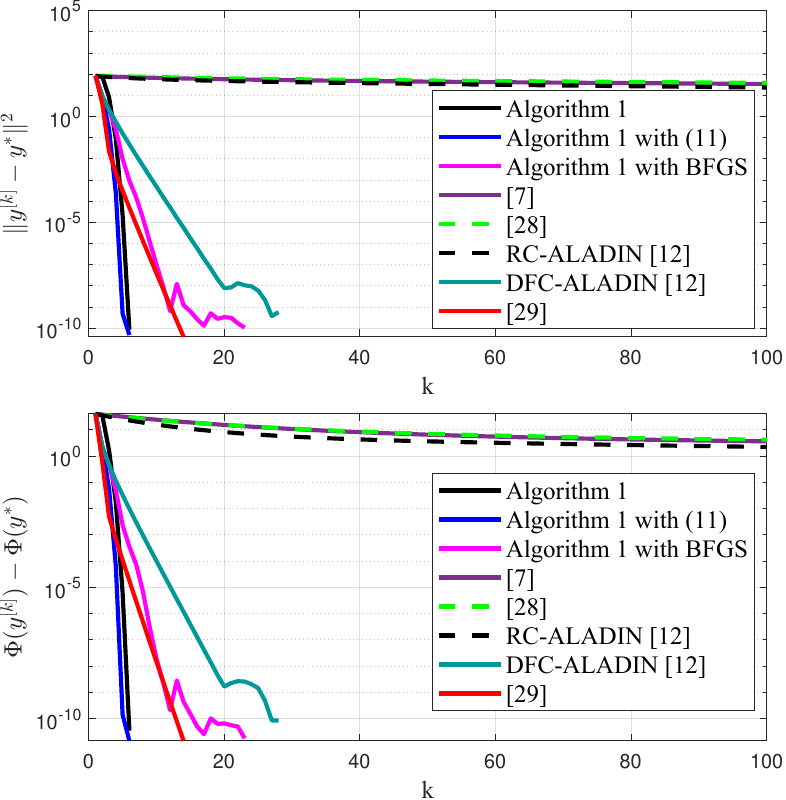}
    \caption{Performance comparison of Algorithm~\ref{alg: AC-ALADIN}, Algorithm~\ref{alg: AC-ALADIN} with \eqref{eq: hessian}, Algorithm~\ref{alg: AC-ALADIN} with BFGS, Consensus ADMM \cite{boyd2011distributed}, DQM\cite{DQM}, RC-ALADIN, DFC-ALADIN \cite{Du2025ACC},  GIANT \cite{wang2018giant}     on the \texttt{ijcnn1.tr} dataset.}
    \label{fig:fig1}
\end{figure}

We evaluate the convergence performance of the proposed Algorithm~\ref{alg: AC-ALADIN} against several baseline methods on the \texttt{ijcnn1} dataset: Consensus ADMM~\cite{boyd2011distributed}, DQM~\cite{DQM}, DFC-ALADIN~\eqref{eq: CALADIN}, RC-ALADIN~\cite{Du2025ACC} (a special case of DFC-ALADIN with $M_i = I$ in \eqref{eq: CALADIN}), and GIANT~\cite{wang2018giant}. The former are based on dual decomposition, while GIANT is a representative primal decomposition based algorithm with superior performance.
 All algorithms are initialized from the same starting point $y^{[1]} = \mathbf{0}$, with dual variables $\lambda_i^{[1]} = \mathbf{0}$. Algorithm~\ref{alg: AC-ALADIN} additionally uses $M_i = I$, $\gamma^{[1]} = N$, and $z^{[1]} = \mathbf{0}$. Consensus ADMM employs first-order primal-dual updates with the augmented Lagrangian parameter $\rho = 1$. DQM adopts a quadratic approximation for local solutions based on Consensus ADMM; it introduces local second‑order approximations but lacks global Hessian information. The original implementation in \cite{DQM} is decentralized; for a fair comparison, we use a version with a coordinator. DFC-ALADIN uses a data‑driven Hessian initialization\footnote{Let $\sigma(t)$ denote the sigmoid function \cite{bishop2006pattern}, and define $\sigma_{ij} := \sigma(b_{ij} a_{ij}^\top y)$. Then the Hessian of $f_i(y)$ can be written as $\nabla^2 f_i(y) = \frac{1}{m} A_i^\top D_i A_i + rI$, where $D_i = \mathrm{diag}(\sigma_{ij}(1 - \sigma_{ij}))$. Since $\sigma(t)(1 - \sigma(t)) \le \frac{1}{4}$ for all $t \in \mathbb{R}$, the Hessian admits the bound $\nabla^2 f_i(y) \preceq \frac{1}{4m} A_i^\top A_i + rI$.}: 
\begin{equation} \label{eq: hessian}\small
M_i = \frac{1}{4m}A_i^\top A_i + rI, 
\end{equation} 
which provides a global upper bound on the exact Hessian of $f_i$, capturing curvature without per‑iteration recomputation. The parameter settings of Algorithm~\ref{alg: AC-ALADIN} are the same as those of RC-ALADIN, except that we additionally set $\epsilon = 10^{-12}$, which is a triggering threshold not present in RC-ALADIN.
Fig.~\ref{fig:fig1} presents the numerical comparison. Consensus ADMM~\cite{boyd2011distributed} exhibits slow convergence due to the absence of second‑order information. DQM~\cite{DQM} employs local second‑order approximations but lacks global curvature information, which limits its speed. RC-ALADIN~\cite{Du2025ACC} shares the same setting as Consensus ADMM, resulting in slow convergence. In contrast, DFC-ALADIN~\eqref{eq: CALADIN} converges faster due to its data‑driven Hessian initialization \eqref{eq: hessian}. Despite lacking dual variables and global Hessian information, GIANT~\cite{wang2018giant} applies exact local Hessian averaging and converges faster than the aforementioned methods in this case. Algorithm 1 (shown as the black solid line) reaches high accuracy in significantly fewer iterations than the competing algorithms. This is attributed to its use of global Hessian information and dual variable updates, which provide accurate curvature approximation and effective coordination. Two enhanced variants further improve performance. The first adopts the Hessian initialization \eqref{eq: hessian} (blue curve), which accelerates convergence. The second employs a Broyden-Fletcher-Goldfarb-Shanno (BFGS)‑based local Hessian reconstruction (cyan solid line) inspired by \cite[Algorithm 2]{Du2025ACC}; it reduces the per‑iteration communication complexity from $\mathcal{O}(Nn^2)$ to $\mathcal{O}(Nn)$ while maintaining favorable convergence (see \cite[Remark 1]{Du2025ACC}). Both variants are consistent with the global convergence results of Section~\ref{sec: convergence}.

Fig.~\ref{fig:fig2} shows the convergence of $\left\|z^{[t+1]} - z^{[t]}\right\|$. The observed sublinear decay is consistent with the theoretical sublinear rate $O\left(\frac{1}{\sqrt{T}}\right)$ established in Section~\ref{sec: convergence}. Moreover, the diminishing updates indicate that Hessian updates are triggered only finitely many times under the sufficient decrease condition, ensuring both efficiency and stability. 
It is worth noting that the discrepancy of the black curve occurs at the first iteration, which is caused by the different Hessian initialization adopted by Algorithm \ref{alg: AC-ALADIN}.
\begin{figure}[htbp]
    \centering
\includegraphics[width=0.52\textwidth,height=0.24\textheight]{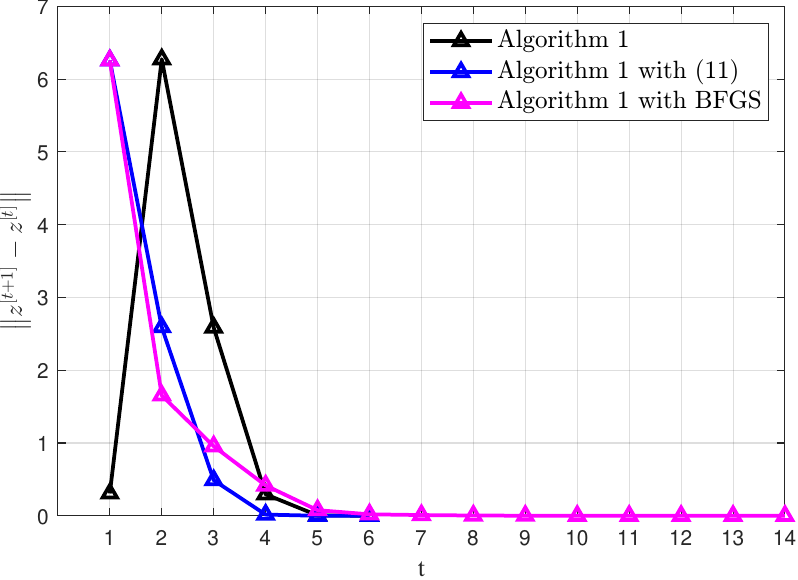}
    \caption{Convergence of $\left\|z^{[t+1]}-z^{[t]} \right\|$ of Algorithm \ref{alg: AC-ALADIN}, Algorithm~\ref{alg: AC-ALADIN} with \eqref{eq: hessian}, Algorithm~\ref{alg: AC-ALADIN} with BFGS.}
    \label{fig:fig2}
\end{figure}

\section{Conclusion}

In this paper, we propose CAPTAIN, a novel distributed optimization algorithm for DCO problems. It integrates DFC‑ALADIN with a finite second‑order information update strategy by introducing an auxiliary variable. We show that the proposed algorithm guarantees global convergence to the optimal solution of convex DCO problems. Moreover, the auxiliary variable exhibits a sublinear convergence rate. Numerical simulations confirm our theoretical findings and demonstrate the efficiency of the algorithm. 

Future work will focus on accelerating CAPTAIN's convergence, reducing its communication and computational costs, and scaling it to high-dimensional problems. We will also extend the method to decentralized settings, e.g., asynchronous updates over directed networks.

\appendices
\section{Proof of Theorem \ref{thm:z_behavior}} \label{app: z_behavior}
% The global convergence of $\{z^{[t]}\}$ sequence can be established in two scenarios. First, if the sufficient decrease condition \eqref{eq: decrease} is never satisfied, then $z^{[t]}$ is never updated. 
% Second, if the condition is satisfied at some iterations, each accepted update of $z^{[t]}$ yields a strictly positive decrease in $\Phi$ (see \eqref{eq: DC}), because $\gamma^{[t]} > 0$ and $\|y^{[k+1]} - z^{[t]}\|^2 > 0$ for any actual update. Since $\Phi$ is bounded below by Assumption~\ref{ass}, an infinite sequence of updates is impossible. Hence, $\{z^{[t]}\}$ is updated only finitely many times. This proves the first statement. 

The global convergence of $\{z^{[t]}\}$ follows from two scenarios.
First, if \eqref{eq: decrease} never holds, the auxiliary variable $z$ is never updated.
Second, if \eqref{eq: decrease} holds at some iteration, the update rule implies $\|y^{[k+1]}-z^{[t]}\|\ge\epsilon$. Consequently,
\begin{equation*}
    \Phi(z^{[t+1]}) = \Phi(y^{[k+1]}) < \Phi(z^{[t]}) - \frac{\gamma^{[t]}}{2}\epsilon^2.
\end{equation*}
By Assumption~\ref{ass}, $M_i\succeq\mu_i I\succ0$, so $\gamma^{[t]}\ge\gamma_{\min}>0$ and $\epsilon>0$ is a fixed positive constant. Hence each accepted update decreases $\Phi$ by at least $\frac{\gamma_{\min}}{2}\epsilon^2>0$.
Since $\Phi$ is bounded below (Assumption~\ref{ass}), only finitely many updates can occur. Thus $\{z^{[t]}\}$ is updated finitely many times, proving the first statement of Theorem~\ref{thm:z_behavior}.

Let $z^{[T]}$ be the last updated auxiliary variable before termination. Three mutually exclusive situations may occur:
\begin{enumerate}
    \item \emph{Final update to the optimum.} If the sufficient decrease condition \eqref{eq: decrease} holds at $y^*$, i.e.,
    \begin{equation}\label{eq: proof_final_update}\small
        \Phi(y^*) < \Phi(z^{[T-1]}) - \frac{\gamma^{[T-1]}}{2} \left\|y^* - z^{[T-1]}\right\|^2,
    \end{equation}
    then the algorithm performs one additional update, sets $z^{[T]} = y^*$, and terminates. In this case, $z^{[T]}$ is the last auxiliary variable.

    \item \emph{Termination due to a small step.} If $\left\|y^{[k+1]} - z^{[T]}\right\| < \epsilon$, then no further update occurs.

    \item \emph{Termination due to violation of the decrease condition.} If the sufficient decrease condition fails, i.e.,
    \begin{equation}\label{eq: proof_final_stop}\small
        \Phi(y^*) \ge \Phi(z^{[T]}) - \frac{\gamma^{[T]}}{2} \left\|y^* - z^{[T]}\right\|^2,
    \end{equation}
    then no further update occurs, and the distance to the optimum satisfies inequality \eqref{eq:z_neighborhood}.
\end{enumerate}
Thus, the second statement of Theorem~\ref{thm:z_behavior} is proved.

% Finally, from the first statement, the sequence $\{z^{[t]}\}$ is updated only finitely many times. Hence, there exists an index $T$ such that for all $t \ge T$, $z^{[t]} = z^{[T]}$, and the Hessian approximations $M_i$ remain fixed thereafter. In this case, Algorithm~\ref{alg: AC-ALADIN} reduces to DFC-ALADIN. Consequently, by Lemma~\ref{lemma 1}, the global variable $y$ and the local variables $x_i$ ($i \in \mathcal V$) generated by the algorithm converge to the global optimum $y^*$ at a linear rate. This completes the proof of the third statement, and thus establishes the global convergence of Algorithm~\ref{alg: AC-ALADIN}.

\section{Proof of Theorem \ref{them: sublinear}} \label{app: sublinear}

In Algorithm~\ref{alg: AC-ALADIN}, whenever the sufficient decrease condition \eqref{eq: decrease} is triggered, we set the auxiliary variable \(z^{[t+1]} = y^{[k+1]}\). This yields the inequality
\begin{equation}\label{eq: decrease2}\small
    \Phi\bigl(z^{[t+1]}\bigr) < \Phi\bigl(z^{[t]}\bigr) - \frac{\gamma^{[t]}}{2}\left\|z^{[t+1]} - z^{[t]}\right\|^2.
\end{equation}
Summing \eqref{eq: decrease2} telescopically for \(t = 1,\dots, T-1\) gives
\begin{equation}\label{eq: decrease3}\small
    \sum_{t=1}^{T-1} \frac{\gamma^{[t]}}{2}\left\|z^{[t+1]} - z^{[t]}\right\|^2 < \Phi\bigl(z^{[1]}\bigr) - \Phi\bigl(z^{[T]}\bigr).
\end{equation}
Since \(0 < \gamma_{\min} \leq \min_t \gamma^{[t]}\), consequently,
\begin{equation}\label{eq: decrease4}\small
    \min_{1\leq t\leq T-1}\left\|z^{[t+1]} - z^{[t]}\right\|^2 < \frac{2}{\gamma_{\min}(T-1)}\Bigl(\Phi\bigl(z^{[1]}\bigr) - \Phi\bigl(z^{[T]}\bigr)\Bigr).
\end{equation}
The right‑hand side of \eqref{eq: decrease4} is finite; therefore \eqref{eq: decrease theorem} holds with \(O\bigl(\frac{1}{\sqrt{T}}\bigr) = \sqrt{\frac{2D}{\gamma_{\min}(T-1)}}\), where \(D \geq \Phi(z^{[1]}) - \Phi(z^{[T]})\). This establishes a sublinear decay rate of the minimum distance between consecutive updates.

\bibliographystyle{IEEEtran}   % 选择合适的参考文献样式
\bibliography{reference}      % 指定你的.bib文件

\end{document}